\documentclass[preprint,12pt]{elsarticle}

\usepackage[utf8]{inputenc}
\usepackage[T1]{fontenc}
\usepackage{lmodern}
\usepackage[english]{babel}
\usepackage[babel=true,final]{microtype}
\usepackage{array}
\usepackage{stmaryrd,amssymb,amsmath,mathtools}
\usepackage{braket}
\usepackage{pgfplots}
\usepackage{xfrac}
\usepackage{enumitem}
\usepackage{siunitx}
\usepackage{etoolbox}
\usepackage{bigfoot}

\newcommand{\ignore}[1]{}

\newcommand{\C}{\mathbb{C}}

\DeclareMathAlphabet{\kw}{\encodingdefault}{\sfdefault}{bx}{n}

\newtheorem{thm}{Theorem}

\newproof{proof}{Proof}

\usepackage{amssymb}

\journal{Nuclear Physics B}

\begin{document}

\begin{frontmatter}



  \title{Checking in Polynomial Time whether or not a Regular Tree
    Language is Deterministic Top-Down}


\author[a1]{Sebastian Maneth}

\affiliation[a1]{organization={University of Bremen, Faculty of Informatics},
            addressline={Bibliothekstraße~5}, 
            city={Bremen},
            postcode={28359}, 
            state={},
            country={Germany}}

\author[a2]{Helmut Seidl}

\affiliation[a2]{organization={Technical University Munich, Institute for Informatics~I2},
            addressline={Boltzmannstr.~3}, 
            city={Garching},
            postcode={85748}, 
            state={},
            country={Germany}}

\begin{abstract}
It is well known that for a given bottom-up tree automaton
it can be decided whether or not there exists deterministic
top-down tree automaton that recognized the same tree language.
Recently it was claimed that such a decision can be carried
out in polynomial time (Leupold and Maneth, FCT'2021);
but their procedure and corresponding property is wrong.
Here we correct this mistake and present a correct property
which allows to determine in polynomial time whether or not
a given tree language can be recognized by a deterministic
top-down tree automaton. Furthermore, our new property is
stated for arbitrary deterministic bottom-up tree automata,
and not for minimal such automata (as before).
\end{abstract}


\end{frontmatter}

\section{Introduction}

Deterministic top-down tree (DTD) automata are strictly weaker than
deterministic bottom-up tree (DBU) automata. In fact, even certain finite
languages cannot be recognized by DTD automata, such as the language
consisting of the two trees $f(a,b)$ and $f(b,a)$.
Despite this weakness, DTD automata have certain advantages over DBU
automata, e.g., they can be evaluated more efficiently.
Several characterizations of DTD languages within the DBU languages
have been studied in the literature.
Viragh~\cite{DBLP:journals/actaC/Viragh81}
introduced the ``path-closed'' property and constructs a so called
``powerset automaton'' for the path closure of a DBU language.
A similar method is presented by G{\'e}csec and
Steinby~\cite{DBLP:books/others/tree1984}.
Another approach is the ``homogeneous closure'' of
Nivat and Podelski~\cite{DBLP:journals/siamcomp/NivatP97};
also they construct an automaton which
has as state set the powerset of the original one.
The exact running time of all these procedures had not been analyzed.

In a recent conference publication by Leupold and
Maneth~\cite{DBLP:conf/fct/LeupoldM21}
a property was presented which supposedly allows to check whether
or not the language given by a minimal deterministic bottom-up tree automaton
can be recognized by a deterministic top-down tree automaton.

Their property was an attempt to lift the well-known
``subtree exchange property'' of Nivat and
Podelski~\cite{DBLP:journals/siamcomp/NivatP97}
to minimal deterministic bottom-up tree automata.
Essentially the property means, that if certain transitions are present,
for instance the transitions
\[
\begin{array}{lcl}
f(q_1,q_2)&\to& q\\
f(q_2,q_1)&\to& q
\end{array}
\]
then also other transitions \emph{must} be present in the given automaton;
in this case the transitions
\[
\begin{array}{lcl}
f(q_1,q_1)&\to& q\\
f(q_2,q_2)&\to& q.
\end{array}
\]

Unfortunately, the property of Leupold and Maneth is \emph{wrong}:
there is a tree automaton which does not satisfy the property, but which
\emph{is} recognizable by a deterministic top-down tree automaton!
Such a counter example was identified and kindly communicated to
the authors by Christof L{\"o}ding. 
The idea of the counter example is as follows.
We consider trees of this form.
\[
\underbrace{g( \cdots g(}_{n\text{ times}} f(x,y) \cdots )
\]
with these properties:
\begin{itemize}
\item
$x,y\in\{a,b\}$ if $n=0$
\item
$x=a$ and $y=b$ if $n>0$ is odd and
\item 
$x=b$ and $y=a$ if $n>0$ is even.
\end{itemize}
The corresponding minimal deterministic bottom-up tree automaton
has the set of states $\{ q_a, q_b, q, p_1, p_2, p, p'\}$.
Its set of final states is $\{q, p_1,p_2,p\}$. 
It consists of the following set of transitions.
\[
\begin{array}{lcllcl}
a&\to&q_a & g(p_1)&\to& p\\
b&\to&q_b & g(p_2)&\to& p'\\
f(q_a,q_a)&\to& q & g(p)&\to& p'\\
f(q_b,q_b)&\to& q & g(p')&\to& p\\
f(q_a,q_b)&\to& p_1\\
f(q_b,q_a)&\to& p_2\ \quad\ &\ 
\end{array}
\]
According to the property of Leupold and Maneth,
the two transitions into the state $q$ demand that also
the transitions
$f(q_a,q_b)\to q$ and 
$f(q_a,q_b)\to q$ be present in the automaton. 
Since they are not present, their property erroneously flags this
language as \emph{not} recognizable by a deterministic top-down
tree automaton. However, as the reader may verify, this language
\emph{can} be recognized by a deterministic top-down tree automaton.
Thus, it is too strict to demand the presence of the above two
transitions into the state $q$, but, it suffices if such transitions
exist into \emph{other} states, as long as these states are
both final states (because $q$ is a final state).
The new property presented in this paper precisely formulates this idea.

Besides being correct, our new property has another advantage:
it is formalized for arbitrary deterministic bottom-up
tree automata, i.e., the automata need \emph{not} be minimal.
Using the example of above, the property states that if there
are transitions
\[
\begin{array}{lcl}
f(q_1,q_1)&\to& q\\
f(q_2,q_2)&\to& q'\\
f(q_1,q_2)&\to& q''
\end{array}
\]
then there \emph{may not} a input context such that starting
from states $q$ and $q'$ the context is accepting, but
starting from $q''$ the context is rejecting. 
I.e., if there was such a context, then the tree language
cannot be recognized by a deterministic top-down tree automaton.

Note that in the context of XML, several notions of
deterministic top-down tree
automata for \emph{unranked} trees have been studied~\cite{DBLP:journals/jcss/GeladeIMNP13,martens,DBLP:journals/tods/MartensNSB06}.
For each of these it takes exponential time to decide whether or
not a given regular unranked tree language is in the respective
class (because regular expressions are used in the definition of
the regular unranked tree language).
Unranked top-down and bottom-up tree automata were already earlier
considered by Cristau, L{\"o}ding,
and Thomas~\cite{DBLP:conf/fct/CristauLT05}.
For their variants they decide (in at least exponential time) whether
a language given by a deterministic bottom-up unranked tree automaton can be
recognized by a deterministic top-down unranked tree automaton. 
Martens, Neven, and Schwentick~\cite{DBLP:conf/birthday/MartensNS08}
present a recent survey of results about deterministic ranked and
unranked tree languages.

\newtheorem{theo}{Theorem}
\newtheorem{prob}{Open Problem}
\newtheorem{prp}[thm]{Proposition}
\newtheorem{lem}[thm]{Lemma}
\newtheorem{cor}[thm]{Corollary}
\newtheorem{ex}[thm]{Example}
\newtheorem{df}[thm]{Definition}
\newcommand{\bex}{\begin{ex}\normalfont}
\newcommand{\bdf}{\begin{df}\normalfont}
\newcommand{\eex}{\qed\end{ex}\normalfont}
\newcommand{\edf}{\end{df}\normalfont}

\section{Preliminaries}

A ranked alphabet $\Sigma$ is a finite set of symbols, each one of which
has associated with it a non-negative integer called it \emph{arity}.
We write $\Sigma=\{f^{(2)}, a^{(0)}, b^{(0)}\}$ to denote that
symbol $f$ has arity two, and $a$ and $b$ both have arity zero.
The subset of symbols of $\Sigma$ of rank $k$ is denoted by
$\Sigma^{(k)}$.

The set $T_\Sigma$ of (ranked, finite, ordered) \emph{trees over} $\Sigma$
is the smallest set $T$ of strings such that whenever
$k\geq 0$, 
$t_1,\dots,t_k\in T$, and
$f\in\Sigma^{(k)}$, then also $f(t_1,\dots,t_k)\in T$.
For trees of the form $a()$ we simply write $a$.

We fix a special ``variable'' symbol $x$ which is not part of
any ranked alphabet. A \emph{context over} $\Sigma$ is a
tree $c$ over $\Sigma$, however, exactly one leaf of $c$ is
labeled by the special symbol $x$. The set of all contexts over $\Sigma$
is denoted by $\C_\Sigma$.
For a tree (or a context) $\xi=f(\xi_1,\dots,\xi_k)$ we define its
\emph{set of nodes} as 
$N(\xi) :=\{ \epsilon\} \cup \{iu \mid i\in\{1,\dots,k\}, u\in N(\xi_i)\}$.
Here $\epsilon$ denotes the root node.

\begin{df}\label{DBA}
A deterministic bottom-up tree automaton (DBA for short) $A$ is
a tuple $(Q,\Sigma,\delta,F)$ where $Q$ is a finite set of states,
$\Sigma$ is a ranked alphabet, $\delta$ is the transition
function, and $F\subseteq Q$ is the set of final states.
The function $\delta$ maps every pair $(w,f)$ with 
$w\in Q^k$, $f\in\Sigma^{(k)}$, and $k\geq 0$ to a state in $Q$.
\end{df}

For $q\in Q$ we denote by $A_q$ the DBA obtained from $A$ by
changing $F$ into the set $\{q\}$.
The transition function of a DBA $A$ is naturally extended
from symbols to trees in $T_\Sigma$:
if $\delta^*(t_i)=q_i$ for $i\in[k]$ and
$\delta(q_1\cdots q_k,f)=q$, then $\delta^*(f(t_1,\dots,t_k))=q$.
The language $L(A)$ of $A$ is defined as
$\{t\in T_\Sigma\mid\delta^*(t)\in F\}$.

We also use the notion of an $A$-run $\rho$: for a tree $t\in T_\Sigma$
it is a mapping $\rho$ from $N(t)$ to $Q$ such that
for all nodes $u\in N(t)$, if $u$ is labeled $f\in\Sigma^{(k)}$,
$\rho(u)=q$ and $\rho(ui)=q_i$ for $i\in[k]$, then it must hold that
$\delta(q_1\cdots q_k,f)=q$. 

\begin{df}\label{DTA}
A deterministic top-down tree automaton (DTA for short) $A$ is
a tuple $(Q,\Sigma,q_0,\delta)$ where $Q$ is a finite set of states,
$\Sigma$ is a ranked alphabet,
$q_\in Q$ is the initial state, and
$\delta$ is the transition function.
The function $\delta$ maps every pair $(q,f)$ with 
$q\in Q$, $f\in\Sigma^{(k)}$, and $k\geq 0$ to a an
element of $Q^k$.
\end{df}

For $q\in Q$ we denote by $A_q$ the DTA obtained from $A$ by
changing $q_0$ into the state $q$.
We define the notion of an $A$-run: for a tree $t\in T_\Sigma$
it is a mapping $\rho$ from $N(t)$ to $Q$ such that
for all nodes $u\in N(t)$, if $u$ is labeled $f\in\Sigma^{(k)}$,
$\rho(u)=q$ and $\rho(ui)=q_i$ for $i\in[k]$, then it must hold that
$\delta(q,f)=q_1\cdots q_k$.

\begin{df}\label{bExchange}
A regular tree language $L$ fulfills the {\em exchange property} if, for every $t \in L$ and every node
$u \in N(t)$,
if $t[u \leftarrow f(t_1 , \dots, t_k )] \in L$ and also $t[u \leftarrow f(s_1 , \dots , s_k )] \in L$, then
$t[u \leftarrow f(t_1 , \dots ,t_{i-1}, s_i, t_{i+1} , \dots , t_k )] \in L$ for each $i = 1, \dots , k$.
\end{df}

\section{Checking if a Tree Language is Deterministic Top-Down}

We say that a tree language $L$ is \emph{deterministic top-down} if there
exists a deterministic top-down tree automaton $A$ such that $L(A)=L$.

Let $A=(Q,\Sigma,\delta,F)$ be a deterministic bottom-up tree automaton (DBA).
A state $q\in Q$ is \emph{reachable} if there exists a tree $t\in T_\Sigma$
such that $\delta^*(t)=q$.
We say that the DBA $A$ is \emph{reduced} if every state $q$ of $A$ is reachable. 
Note that for a given DBA it is straightforward to construct in polynomial
time an equivalent DBA that is reduced. 
We therefore assume from now on that each given DBA is reduced.

\begin{df}\label{df:conflict}
Let $A=(Q,\Sigma,\delta,F)$ be a DBA.
%
%
The triple $(q,q',q'')\in Q^3$ is called a \emph{conflict} 
if there is an input symbol $f\in\Sigma$ of rank $k\geq 2$,
a context $c\in\C_\Sigma$, 
states $p_1,\ldots,p_k,p'_1,\ldots,p'_k\in Q$,
and an index $j$ with $1\leq j\leq k$
such that the following holds:
\begin{enumerate}
\item[(1)] $\delta(p_1\ldots p_k,f) = q$,
\item[(2)] $\delta(p'_1\ldots p'_k,f) = q'$,
\item[(3)] $\delta(p_1\ldots p_{j-1}p'_jp_{j+1}\ldots p_k,f) = q''$,
\item[(4)] $\delta^*(q,c)\in F$ as well as $\delta^*(q',c)\in F$, but
\item[(5)] $\delta^*(q'',c)\not\in F$.
\end{enumerate}
\end{df}

Note that it follows from the fact that $A$ is a \emph{deterministic}
bottom-up tree automaton, that $q\not=q''$ and $q'\not=q''$.

\begin{lem}\label{lem:no_conflict}
If a BTA $A$ has conflicts, then $L(A)$ is not deterministic top-down.
\end{lem}
\begin{proof}
Let $A=(Q,\Sigma,\delta,F)$. Since $A$ has a conflict, there are 
$(q,q',q'')\in Q^3$, 
$f\in\Sigma^{(k)}$ with $k\geq 2$,
$c\in\C_\Sigma$,
$j\in[k]$, and  
$p_1,\ldots,p_k,p'_1,\ldots,p'_k\in Q$ such that 
Conditions~(1)--(5) of Definition~\ref{df:conflict} hold.
Let $t_1,\dots,t_k,t_1',\dots,t_k'\in T_\Sigma$ such that, for $i\in[k]$,
$\delta^*(t_i)=p_i$ and $\delta^*(t_i')=p_i'$.
From Condition~(4) we know that 
both $c[f(t_1,\dots, t_k)]$ and $c[f(t_1',\dots, t_k')]$ are in $L(A)$.
Assume that $L(A)$ is deterministic top-down.
Then by the subtree exchange property, also 
$t=c[f(t_1,\dots, t_{j-1}, t'_{j}, t_{j+1}, t_k)]\in L(A)$.
However, by Conditions~(3) and~(5), $t\not\in L(A)$. Thus, $L(A)$ is not deterministic top-down.
\qed
\end{proof}

Recall that DTA abbreviates ``deterministic top-down tree automaton''.
For a given DBA $A$ we now define its ``associated DTA''.
This associated DTA always accepts a superset of $L(A)$.
Later we will show that if $A$ has no conflicts (and is ``complete''), then
in fact the associated DTA accepts precisely $L(A)$. 
Let $A=(Q,\Sigma,\delta,F)$.
The \emph{associated DTA} of $A$ is the DTA $A'=(2^Q,\Sigma, F, \delta')$.
Let $Q_0$ be a non-empty subset of $Q$.
Let $f\in \Sigma^{(k)}$ with $k\geq 1$.
Then let
\[
\delta'(Q_0,f)= Q_1\cdots Q_k
\]
where the $Q_i$ are defined as follows.
Consider all $f$-transitions of $A$ which have a right-hand side in $Q_0$
and denote their left-hand sides by 
\[
f(q^1_1,\dots,q^1_k),\dots, f(q^m_1,\dots,q^m_k).
\]
Then $Q_i=\{q^1_i,\dots, q^m_i\}$ for $i\in[k]$.
If no such transitions exist then
let $Q_i=\emptyset$ for $i\in[k]$.
Finally, for $a\in\Sigma^{(0)}$, if there is an $a$-transition
of $A$ which has a right hand side in $Q_0$, then let
\[
\delta'(Q_0,a)=\varepsilon.
\]

\begin{lem}\label{lem:associated}
Let $A$ be a DBA and let $A'$ be the DTA associated to $A$.
Then $L(A)\subseteq L(A')$.
\end{lem}
\begin{proof}
Let $A=(Q,\Sigma,\delta,F)$.
The desired inclusion follows from Claim~1:
if $t\in T_\Sigma$ is in $L(A)$, then there is an $A$-run $r$ on $t$ with $q=r(\varepsilon)\in F$.
By Claim~1, there is an $A'$-run on $t$ starting in any state
of $A'$ containing $q$; since the start state of $A'$ contains $q$, we obtain
that $t\in L(A')$.

\medskip

\textbf{Claim 1:}\quad
Let $t\in T_\Sigma$ and let $r$ be an $A$-run on $t$.
Let $Q_0\in 2^Q$ with $r(\varepsilon)\in Q_0$.
There exists an $A'$-run $r'$ on $t$ such that $r'(\varepsilon)=Q_0$
and $r(u)\in r'(u)$ for all $u\in V(t)$.

\medskip

The claim is proved by induction on the structure of $t$.
Let $t=a\in\Sigma^{(0)}$.
If $r(\varepsilon)=q$ then $\delta(a, \varepsilon)=q$; by the definition of $A$,
$\delta'(Q_0',a)=\varepsilon$ for every $Q_0'\in 2^Q$ such that
$q\in Q_0'$.
Hence $\delta'(Q_0,a)=\varepsilon$ which implies
the existence of an $A'$-run $r'$ with $r'(\varepsilon)=Q_0$.

Let $t=f(t_1,\dots,t_k)$ with $k\geq 1$. By the definition of $A'$, there
is a transition $\delta'(Q_0,f)=Q_1\cdots Q_k$ such that
$r(i)\in Q_i$ for $i\in[k]$; this is because $A$ has an $f$-transition
with left-hand side $(f,r(1)\dots r(k))$ and right-hand side $r(\varepsilon)\in Q_0$.
Thus, there is an $A'$-run $r'$ on $t$ with $r'(\varepsilon)=Q_0$.
By induction there are $A'$-runs $r_i'$ on $t_i$ with
$r_i'(\varepsilon)=Q_i$ for $i\in[k]$. Thus, $r'$ exists with
$r'(iu)=r_i'(u)$ for $i\in[k]$.
\qed
\end{proof}

\begin{ex}
Let $A=(\{q,  q_a, q_b, p, p', p_{ab}, p_{ba}\}, \Sigma,\delta, \{q, p, p_{ab}, p_{ba} \})$ be a DBA with the transition function defined as follows:
\begin{align*}
	\delta(a,\varepsilon) &= q_a, &
	\delta(b,\varepsilon) &= q_b,\\
	\delta(f,q_aq_a) &= q, &
	\delta(f,q_bq_b) &= q,\\
	\delta(f,q_aq_b) &= p_{ab}, &
	\delta(f,q_bq_a) &= p_{ba},\\
	\delta(g,p_{ab}) &= p,   &
	\delta(g,p_{ba}) &= p',\\
	\delta(g,p') &= p,   &
	\delta(g,p) &= p',	 
\end{align*}

We apply the
construction described preceding Lemma~\ref{lem:associated} to $A$.
The set of final states is $\{q, p, p_{ab}, p_{ba} \}$.
This set is now the initial state of our DTA $A'$.
We only discuss the states that are reachable from
the initial state. 
Let us consider the input symbol $g$.
We now collect the left-hand sides of all $g$-transitions of $A$
which have a right-hand side in $F$:
$(g, p_{ab})$ and $(g, p')$. Thus, $\{ p_{ab},p'\}$ is a state of $A$ and
\[
\delta'(\{q, p, p_{ab}, p_{ba} \},g)= \{p_1,p'\}
\]
is a transition of $A'$.
In the next step, we collect the left-hand sides of all $g$-transitions
of $A$ which have a right-hand side in $\{p_1,p'\}$:
$(g,p_{ba})$ and $(g,p)$. Thus $\{p_{ba},p\}$ is a state.
We obtain these two $g$-transitions of $A'$:
\[
\begin{array}{lcl}
\delta'(\{p_{ab},p'\}, g)&=& \{p_{ba},p\} \\
\delta'(\{p_{ba},p\}, g)&=& \{p_{ab},p'\}.
\end{array}
\]
Next we consider the input symbol $f$.
The left-hand sides of $f$-transitions
of $A$ which have their right-hand side in the final state $\{q, p, p_{ab}, p_{ba} \}$ are
$(f,q_aq_a), (f,q_bq_b), (f,q_aq_b)$, and $(f,q_bq_a)$.
Thus 
\[
\delta'(\{q, p, p_{ab}, p_{ba} \},f)= \{q_a,q_b\}\{q_a,q_b\}
\]
is a transition of $A'$.
For the states $\{p_{ab},p'\}$ and $\{p_{ba},p\}$ we add the transitions
\[
\begin{array}{lcl}
\delta'(\{p_{ab},p'\},f)&=& \{p_a\},\{q_b\} \\
\delta'(\{p_{ba},p\},f)&=& \{q_b\},\{q_a\}
\end{array}
\]
to $A'$.
Finally, for $x\in \{a,b\}$ we add
$\delta'(\{q_a,q_b\},x) = \varepsilon $ and
$\delta'(\{ q_x \},x) =  \varepsilon $ to $A'$.
\end{ex}

Let $A=(Q,\Sigma,\delta,F)$ be a DBA and let $q_{\text{trap}}$ be
a fresh symbol not in $Q$.
The \emph{completion} of $A$ is the DBA $A'=(Q\cup\{q_{\text{trap}}\},\Sigma,\delta\cup \delta',F)$
where $\delta'(w,f)=q_{\text{trap}}$ whenever
$\delta(w,f)$ is undefined for $f\in\Sigma^{(k)}$, $k\geq 0$, and $w\in Q^m$.
A \emph{completed DBA} is the completion of some DBA $A$. 

\begin{lem}\label{lem:conflict}
Let $A$ be a completed DTA that has no conflicts and 
let $A'$ be the associated DTA of $A$.
Then $L(A')=L(A)$.
\end{lem}
\begin{proof}
From Lemma~\ref{lem:associated} we already know that $L(A')\supseteq L(A)$.
Hence, it remains to show that $L(A')\subseteq L(A)$.
This follows from Claim~2 for $Q_0=F$ (and taking $c=x_1$).

\medskip

\textbf{Claim 2:}\quad
Let $Q_0\subseteq 2^Q-\{\emptyset\}$ such that there exists a context $c\in\C_\Sigma$ with
(1)~for all $q\in Q_0$: $\delta^*(q,c)\in F$ and
(2)~for all $q\in Q-Q_0$: $\delta^*(q,c)\not\in F$.
Let $t\in T_\Sigma$.
If $t\in L(A_{Q_0}')$ then $t\in\bigcup_{q\in Q_0}L(A_q)$.

\medskip

The claim is proved by induction on the structure of $t$.
Let $t=a\in\Sigma^{(0)}$.
If $t\in L(A_{Q_0}')$ then by the definition of $A'$ there exists a $q\in Q_0$
such that $\delta(a,\varepsilon)=q$, i.e., $t\in L(A_q)$.

Let $t=f(t_1,\dots,t_k)$ for $k\geq 1$.
Let $Q_1,\dots,Q_k$ such that $\delta'(Q_0,f)=Q_1\cdots Q_k$.
Assume now that $t\in L(A_{Q_0}')$.
By the definition of $A'$ this implies that $Q_i\not=\emptyset$ for every $i\in[k]$.
We now want to establish that we can apply the induction hypothesis to
each $t_i$ with $i\in[k]$.
Let $i\in[k]$ and $q_i\in Q_i$.
By the definition of $A$ there are $q_j$ with $j\in[k]-\{i\}$ and 
the transition 
$\delta(f,q_1\cdots q_k)\in Q_0$.
Let $c_i=c[f(\tilde{t}_1,\dots,\tilde{t}_{i-1},x_1,\tilde{t}_{i+1},\dots,\tilde{t}_k)$.
where $c$ is a context as in the claim and
each $\tilde{t}_j$ is an arbitrary tree in $L(A_{q_j})$ for
$j\in[k]-\{i\}$.
Such trees exist because $A$ is reduced.
Thus $\delta^*(q_i,c_i)\in F$.
Consider any other $q_i'\in Q_i$.
Then there are $q'_i$ such that 
$\delta(f,q_1'\cdots q_k')\in Q_0$.
Since there is no conflict,
$\delta(f,q_1\cdots q_{i-1} q_i' q_{i+1} \cdots q_k)=q'$
and $\delta^*(c,q')\in F$.
Thus $\delta^*(q_i',c_i)\in F$.
Now consider some $q_i'\in Q-Q_i$.
By the definition of $A$ there is no transition
$\delta(f,p_1\cdots p_{i-1} q_i' p_{i+1}\cdots p_k)\in Q_0$ for
any $p_1,\dots,p_k\in Q$. Since $A$ is complete, this implies that
$\delta(f,p_1\cdots p_{i-1} q_i' p_{i+1}\cdots p_k)\in Q-Q_0$ for
any $p_1,\dots,p_k\in Q$.

Thus, we can apply the induction hypothesis to $Q_i$ and $t_i$
for $i\in[k]$.
we obtain that $t_i\in\bigcup_{q\in Q_i} L(A_{q_i})$.
Let $q_i$ be the unique state in $Q_i$ such that $\delta^*(t_i)=q_i$.
By the definition of $A$, $\delta(f,q_1\cdots q_k)\in Q_0$.
Thus $t\in\bigcup_{q\in Q_0}L(A_q)$.
\qed
\end{proof}

From Lemmas~\ref{lem:no_conflict} and \ref{lem:conflict} we obtain the following theorem.

\begin{thm}
The language $L(A)$ of a
completed DBA $A$ is deterministic top-down iff $A$ has no conflicts.
\end{thm}  

We call a state $q_t$ a \emph{trap} state of $A$, if
$q_t\not\in F$ and for all $f\in\Sigma^{(k)}$ and $q_1,\ldots,q_k\in Q$,
$\delta(q_1\ldots q_k,f) = q_t$ whenever $q_i=q_t$ for some $i\in[k]$.
Since transitions containing trap states need not explicitly represented,
we define the \emph{size} $|A|$ of the automaton $A$ as
\[
	|A| = |Q|+\sum\{k+1\mid \delta(q_1\ldots q_k,f)\neq q_t\}
\]
where we w.l.o.g.\ assume that all input symbols of $A$ also occur in
transitions without trap states.
We now show that testing whether or not a given DBA $A$ has conflicts
can be achieved in polynomial time.

\begin{thm}
Let $A=(Q,\Sigma,\delta,F)$ be a deterministic bottom-up tree automaton.
Let $n=|Q|$, $m=|\delta|$, and let $a$ be the maximal arity
of the symbols in $\Sigma$.
It is decidable in time $O(n^3 m^2 a)$ whether or not $L(A)$
can be recognized by a deterministic top-down automaton.
\end{thm}
\begin{proof}
In order to verify whether or not $A$ has conflicts, we first construct
the set $T_0$ of all triples satisfying the first three properties.
This takes time $O(m^2k)$.

Then we construct the set $T$ of triples of states reachable from $T_0$ by means of some context.
Technically, the set $T$ is inductively defined by
\begin{itemize}
\item	$T_0\subseteq T$;
\item	Assume that $(q,q',q'')\in T$. Then
$(q_1,q'_1,q''_1)\in T$ if  there is some input symbol $f\in\Sigma$ of arity $k\geq 1$
some index $j$ with $1\leq j\leq k$ together with states
$p_1,\ldots, p_{j-1},p_{j+1},\ldots,p_k\in Q$ such that
\begin{itemize}
\item $\delta(p_1\ldots p_{j-1}qp_{j+1}\ldots p_k, f) = q_1$,  
\item $\delta(p_1\ldots p_{j-1}q'p_{j+1}\ldots p_k, f) = q'_1$, and
\item $\delta(p_1\ldots p_{j-1}q''p_{j+1}\ldots p_k, f) = q''_1$.
\end{itemize}
\end{itemize}
Indeed, determining $T$ from $T$ can be done by repeatedly adding new triples.
	To find the next triple from a given triple $(q,q',q'')$, at most all transtions in $\delta$ without trap state need 
	to be consulted and for each such transition each position $i$ where the state $q$ may occur.
	Accordingly, the set $T$ can be constructed in time $O(n^3\cdot m\cdot a)$.

Then $A$ has a conflict iff $T\cap F\times F\times(Q\backslash F)\neq\emptyset$.
	Given the set $T$, this check can be implemented in time $O(n^3)$.
	Thus, altogether the algorithm requires time $O(m^2\cdot a+ n^3\cdot m\cdot a)\leq O(n^3\cdot m^2\cdot a)$.
        \qed
\end{proof}


\bibliographystyle{abbrv}
\bibliography{lit}

\end{document}